\documentclass[12pt]{article}

\usepackage{fullpage}

\usepackage{graphicx}
\usepackage{amsfonts}
\usepackage{amscd}
\usepackage{indentfirst}
\usepackage{amssymb,amsthm, amsgen, amstext, amsbsy, amsopn,yfonts}
\usepackage{xcolor}
\usepackage{amsmath}
\usepackage{mathrsfs}
\usepackage{ulem} %

\newtheorem{thm}{Theorem}[section]

\newtheorem{prop}[thm]{Proposition}

\newtheorem{rem}[thm]{Remark}

\newtheorem{problem}[thm]{Problem}

\newcommand{\R}{{\mathbb{R}}}

\newcommand{\glaub}{\rm glaub}
\newcommand{\cont}{\rm cont}
\newcommand{\reg}{\rm reg}
\newcommand{\front}{\rm front} 
\newcommand{\lag}{\rm lag}     

\begin{document}

\title{
  Contact geometric approach to Glauber dynamics \\
  near a cusp and its limitation
}
\author{\textsc
  Shin-itiro Goto\thanks{ 
    Center for Mathematical Science and Artificial Intelligence,
    Chubu University,
    1200 Matsumoto-cho, Kasugai, Aichi 487-8501, Japan},\quad
  Shai Lerer\thanks{School of Mathematical Sciences,
    Tel Aviv University, 6997801, Tel Aviv, Israel},\quad
  Leonid Polterovich\footnotemark[2]
}

\date{\today}

\maketitle
\abstract{
  We study a nonequilibrium  mean field Ising model
  in the low temperature
  phase regime, where metastable equilibrium states develop a
  cuspidal
(spinodal) singularity. We focus on celebrated Glauber dynamics,
and design a contact Hamiltonian flow which captures
some of its rough features in this regime. We prove, however, that there is an inevitable discrepancy between the scaling laws for the relaxation time in the
Glauber and the contact Hamiltonian dynamical systems.
}

\section{Introduction}

Considerable activity is being devoted to  establish
a solid foundation of nonequilibrium thermodynamics and statistical mechanics
\cite{Kubo,ZubarevI,ZubarevII}.
Among various problems in constructing a viable general theory, establishing a concise description of dynamical properties of thermodynamic systems with phase transitions is one of the main points. Recalling the success with the Ising model
in developing equilibrium statistical mechanics, one recognizes that analyzing a canonical dynamical model is expected to be the first step towards the establishment of a nonequilibrium theory.
One can choose a spin kinetic model as a canonical model, where
its dynamics is called Glauber dynamics.
In particular the model with the mean field type spin coupling
enables one to derive a simple dynamical system (see equation \eqref{eq-kubo} below) for an expectation
or thermal average
of magnetization with some approximation  \cite{Glauber,Suzuki-Kubo}.
This expectation variable is
the thermodynamic conjugate variable of the externally applied magnetic field,
and shows relaxation processes.
Here, roughly speaking, relaxation is
a dynamical process starting from a nonequilibrium state
to a point of the equilibrium state set in thermodynamic phase space.
At equilibrium,
for the Ising model with mean field type
interactions, the equation of state is explicitly derived
by calculating the partition function in the thermodynamic limit, and
the system at equilibrium exhibits a first-order phase transition together with metastable states (see \cite{Goto} for this derivation).
For Glauber dynamics, various scaling relations have been proposed, and
one of them
 is the scaling of the relaxation time near the critical point \cite{A,L}.

To advance our understanding of Glauber dynamics, one may employ reliable and well-developed mathematical theories. One of them is contact geometry whose central object is the Gibbs $1$-form $dz-pdq$ in the 3-dimensional thermodynamic phase space equipped with coordinates $z$ (minus free energy), $p$ (magnetization), and
$q$ (exterior magnetic field) \cite{EP}. The equilibrium
submanifold of the mean field Ising model is represented by a smooth curve. The restriction of the
Gibbs form to the equilibrium curve vanishes, which manifests the fundamental thermodynamic relation. In the presence of the 1st order phase transition, the equilibrium curves necessarily develop singularities when projected to
 the $(z,q)$-plane. We are especially interested in so-called spinodal points where such a singularity is a cusp. The interplay between scaling relations for the relaxation time and the geometry of the equilibrium curve near a spinodal point is the main theme of the present paper.

Another merit of contact geometry is that it provides a natural class of dynamical systems,
so called contact Hamiltonian flows on the thermodynamic phase space, which preserve the Gibbs form
up to a conformal factor.  Loosely speaking, contact Hamiltonian flows are odd-dimensional cousins
 of the standard Hamiltonian flows of classical mechanics. Contact Hamiltonian flows model processes of nonequilibrium thermodynamics
 (see e.g. \cite{Haslach,Grmela,Goto0,Bravetti,V,EP,Goto}).
 In this paper, we focus on designing a contact Hamiltonian system whose dynamics  captures  the equilibria and their stability patterns of the Glauber-Suzuki-Kubo ordinary differential equation (ODE) \eqref{eq-kubo}  near critical points, show
 time-scales near critical points for the phase transition,
 and compare this with other proposals in the literature.
 In addition to contact geometry, we use some basics of
 singularity theory.
 Because of these,
   commonly used notations employed in contact geometry
   are adopted in this paper.

\section{Mean field Ising model}
Thermodynamics of the mean field Ising model in the presence of a  
constant magnetic field
$q$ is described by its free energy
(taken with the opposite sign) $z$,  and magnetization
$p$, where $z$ and $p$ are obtained by dividing by
  the number of total spins $N$.
  In obtaining the thermodynamic quantities $z$ and $p$
  from the microscopic model
  with the standard statistical method, the limit $N\gg 1$ has been taken.
  In addition,
  the statistical average over spin variables is assumed to yield a
  proper scaling of $N$ so that the existence of the thermodynamic limit is guaranteed.
 These involved
  variables can be written in terms of
  commonly used notations in physics as shown in Table \ref{table-notation}.
\begin{table}[bht]
  \begin{center}
  \caption{Notations of various quantities}
  \label{table-notation}
  \begin{tabular}{|c||c|c|}
      \hline
      Quantity & Geometry oriented symbol & Physics oriented symbol\\
      \hline
      Magnetization &$p$ & $m$\\
      \hline
      Magnetic field &$q$&$h$\\
      \hline
      Interaction &$b$&$J$\\
      \hline
      Free energy &$-z$&$f$\\
      \hline
      \end{tabular}
  \end{center}
\end{table}

In equilibrium,
we have the relations
\begin{equation}\label{eq-eqstate}
p=\phi'(q+bp)\;,\qquad
z=\phi(q+bp) - \frac{b}{2}p^{2}\;,
\end{equation}
where
$\phi(u) = \beta^{-1}\ln\left( 2\cosh (\beta u)\right)$ and
the first equation represents the so-called self-consistent equation (see e.g. \cite{Goto}).
The real parameter $\beta >0 $ is the inverse temperature, and $b>0$ is determined by the strength of the interaction and the geometry of the model.
Equations \eqref{eq-eqstate} can be resolved as
$$
q(p) = -bp+ \frac{1}{2\beta}\ln \frac{1+p}{1-p}\;,
$$
and
$$
z(p) = \frac{1}{\beta}\ln 2 - \frac{1}{2\beta}\ln (1-p^2) - \frac{bp^2}{2}\;.
$$
Note that
$$
-z(p)
= - \frac{bp^2}{2}
- pq(p)
+ \frac{1+p}{2\beta}\ln(1+p)
+ \frac{1-p}{2\beta}\ln (1-p)
- \frac{1}{\beta} \ln 2\;,
$$
which is another expression for the free energy of
the mean field Ising model \cite[formula (13.1.14)]{BH},\cite[formula (1)]{ME}.

Consider the curve $L = \{(q(p),p)\}$, $p \in (-1,1)$ in the $(q,p)$-plane
given by the first equation in \eqref{eq-eqstate}.
The point $(q_*, p_*)$  on $L$ with $q_* = q(p_*), dq/dp  (p_*) =0$
is called the {\it spinodal}. Spinodal points exist when
\begin{equation}\label{eq-bbeta}
b\beta > 1\;,
\end{equation}
in which case the value of $p_*$ is given by
\begin{equation}\label{eq-spinod-vsp}
p_* = \pm \sqrt{1-\frac{1}{b\beta}}\;,
\end{equation}
see Figure \ref{ising-fig1}.
In what follows, without loss of generality we choose the
plus sign and put $q_*= q(p_*)$.
The explicit expression of $q_{*}$ in terms of  
$b$ and $\beta$ is
  $$
  q_{*}
  =-b\sqrt{\frac{b\beta-1}{b\beta}}+\frac{1}{\beta}\text{arctanh}
  \sqrt{\frac{b\beta-1}{b\beta}}.
  $$
Note that $q_* < 0$.
Since $dq/dp$ = $1/(dp/dq)$, spinodal points are where $dp/dq$ diverges, and they are physically interpreted as the points where a response of $p$
(magnetization) due to a change of $q$ (exterior magnetic field)
diverges. 

\medskip
\begin{figure}[h]
\label{ising-fig1}
\includegraphics[bb=0.0 0.0 1439.820022 623.922010,width=1.0\hsize]{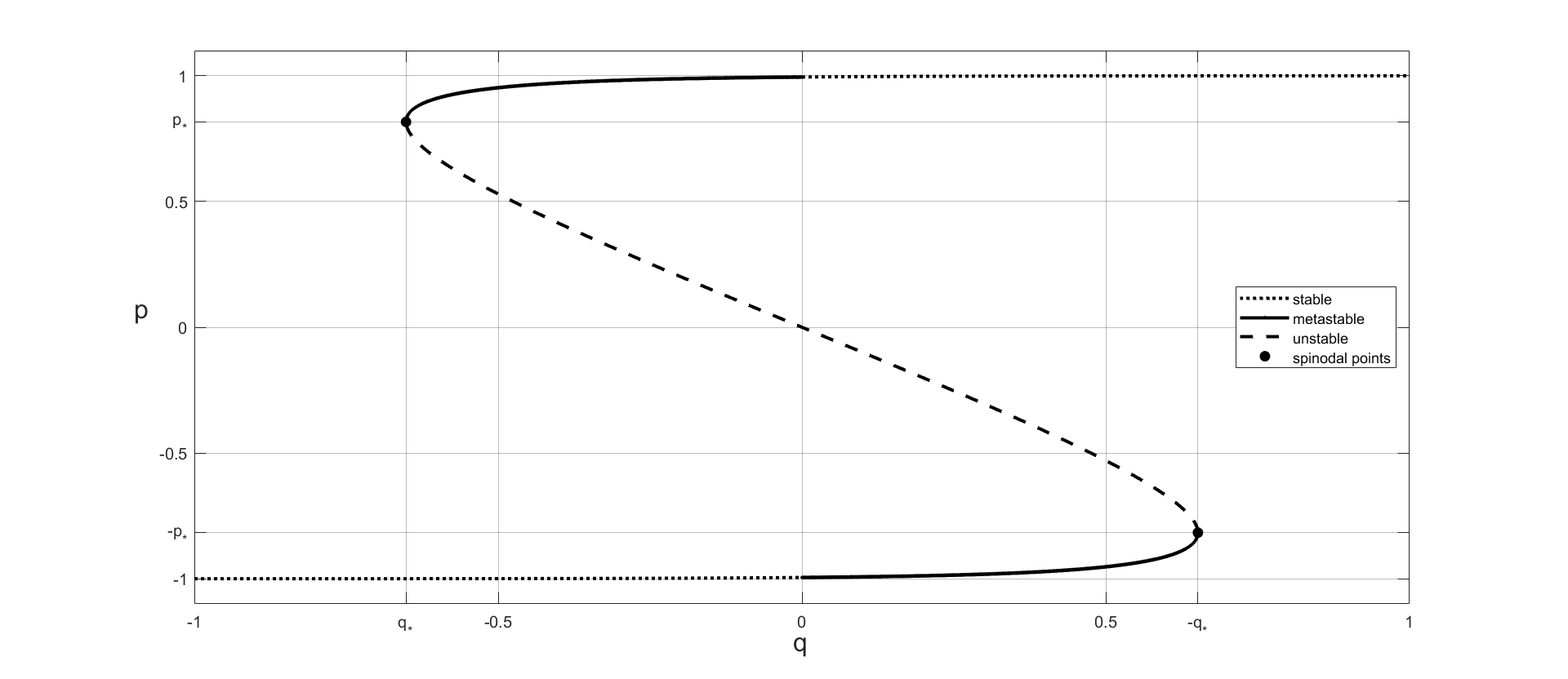}
\centering
\caption{ Lagrangian  projection of $\Lambda$ to the $(q,p)$-plane.}
\end{figure}

In the present paper we study nonequilibrium thermodynamics of the mean field Ising model in a small neighbourhood of the spinodal point. We deal with relaxation processes where the magnetic field $q$ is constant, while the magnetization $p(t)$ converges to a limit $p_{\infty}$ as $t \to \infty$. For $q \neq q_*$, we shall prove existence of the limit
\begin{equation}\label{eq-relax}
-\tau^{-1}:= \lim_{t \to +\infty} \frac{\ln |p(t)-p_{\infty}|}{t}\;.
\end{equation}
When $\tau>0$,  we have an exponential convergence of $p(t)$ to its equilibrium value. 
In this case $\tau$ is called the {\it relaxation time}. We shall focus on $\tau$ as
a function of the magnetic field in two models of nonequilibrium thermodynamics.

The first model, a classical one, is Glauber dynamics.
In \cite{Glauber} Glauber described a Markov process which models relaxation
of the Ising chain to the equilibrium. Furthermore, Glauber \cite{Glauber} and later Suzuki and Kubo
\cite{Suzuki-Kubo} proposed an ordinary differential equation
\begin{equation}\label{eq-kubo}
\dot{p} = -p + \phi'(q+bp)\;,
\end{equation}
where $q$ is constant, see equation (4.3) in  \cite{Suzuki-Kubo}, which provides a molecular field
approximation to the Markov evolution in the thermodynamic limit.
One of the features of this approximation is as follows.
Consider the regime when $b\beta >1$, cf. \eqref{eq-bbeta}.
In this regime, for $q$ from the interval $(q_*,-q_*)$ equation \eqref{eq-kubo} has three equilibrium points:
the maximal and the minimal equilibria are stable, and the one in the middle is unstable.
The stable equilibrium having bigger free energy
(i.e., the smaller value of $z$)
corresponds to the metastable equilibrium of the Markov process. 
Metastability, roughly speaking, means that in the thermodynamic limit, i.e., as the size of the chain increases, the chain spends larger and larger time in the metastable region, see Theorem 13.1 in \cite{BH} for a precise formulation.  
Metastable equilibria are represented
by a solid line on Figure \ref{ising-fig1}.

\medskip\noindent
{\bf Convention:} In this paper, by  
    Glauber dynamics we mean the dynamics of
ODE \eqref{eq-kubo}. By the {\it metastable} equilibrium of  \eqref{eq-kubo}
we mean the dynamically stable equilibrium with the smaller value of $z$.
Let us mention also that equation (4.3) in \cite{Suzuki-Kubo} contains 
a multiplicative
factor responsible for the time units. 
We omit it, tacitly assuming it to be $1$, 
for the sake of simplicity of the notation. 

\medskip
The second model of relaxation processes which we consider is the contact Hamiltonian dynamics
in the thermodynamic phase space. The thermodynamic phase space $\R^3$ is equipped with the coordinates
$z$ (free energy taken with the opposite sign), $p$ (magnetization), and $q$ (magnetic field). The contact form is given by $dz-pdq$. The equilibrium Legendrian submanifold $\Lambda$ is given by
equations \eqref{eq-eqstate}. We shall consider the projections
$$
\zeta_{\front}:\R^3 \to \R^2, \quad 
(z,q,p) \mapsto (z,q)
$$
and
$$
\zeta_{\lag}: \R^3 \to \R^2, \quad 
(z,q,p) \mapsto (q,p)
$$
to the $(z,q)$- and $(q,p)$- planes, respectively.
The first projection
is called the {\it front projection}, and the image $\zeta_{\front}(\Lambda)$ is called the front and is denoted by $\Sigma$. It plays a crucial role 
in future considerations. The second projection is called the
{\it Lagrangian projection},  and we denote
$L = \zeta_{\lag}(\Lambda)$. We have encountered earlier this curve
when we introduced the spinodal points $(\pm q_*,p_*) \in L$.
Let $(z_*,\pm q_*,p_*)$ be the lift of $(\pm q_*,p_*)$ to $\Lambda$.
Then the points
$$
C = \zeta_{\lag}( ( z_*, q_*,p_*)) = (z_*,q_*)
$$
and
$$
C' = \zeta_{\lag}( ( z_*,-q_*,p_*)) = (z_*,-q_*)
$$
correspond to  the left and  the right cusps of the front, respectively 
(see Figure \ref{ising-fig2}). 
  Note that $z_* = z(p_*)$. Abusing the language, we also call $C$ and $C'$
spinodal points.
 Another singularity of the front is its double point $D= (z^*,0)$;
 it will be ignored in the present paper. Write
 $\Sigma_{\reg}= \Sigma \setminus \{C,C', D\}$ for the regular part of $\Sigma$.
Dotted, solid, and dashed lines 
on $\Sigma_{\reg}$ are denoted by $S_\pm$, $M_\pm$ and $U$; they correspond to stable, metastable and unstable equilibria of Glauber dynamics,
respectively.
In the present paper we focus on a small neighbourhood of the spinodal point $C$, see the shaded region on Figure \ref{ising-fig2}

\medskip
\begin{figure}[h]
  \label{ising-fig2}
\includegraphics[bb=0.0 0.0 1920.0 794.0,width=1.0\hsize]{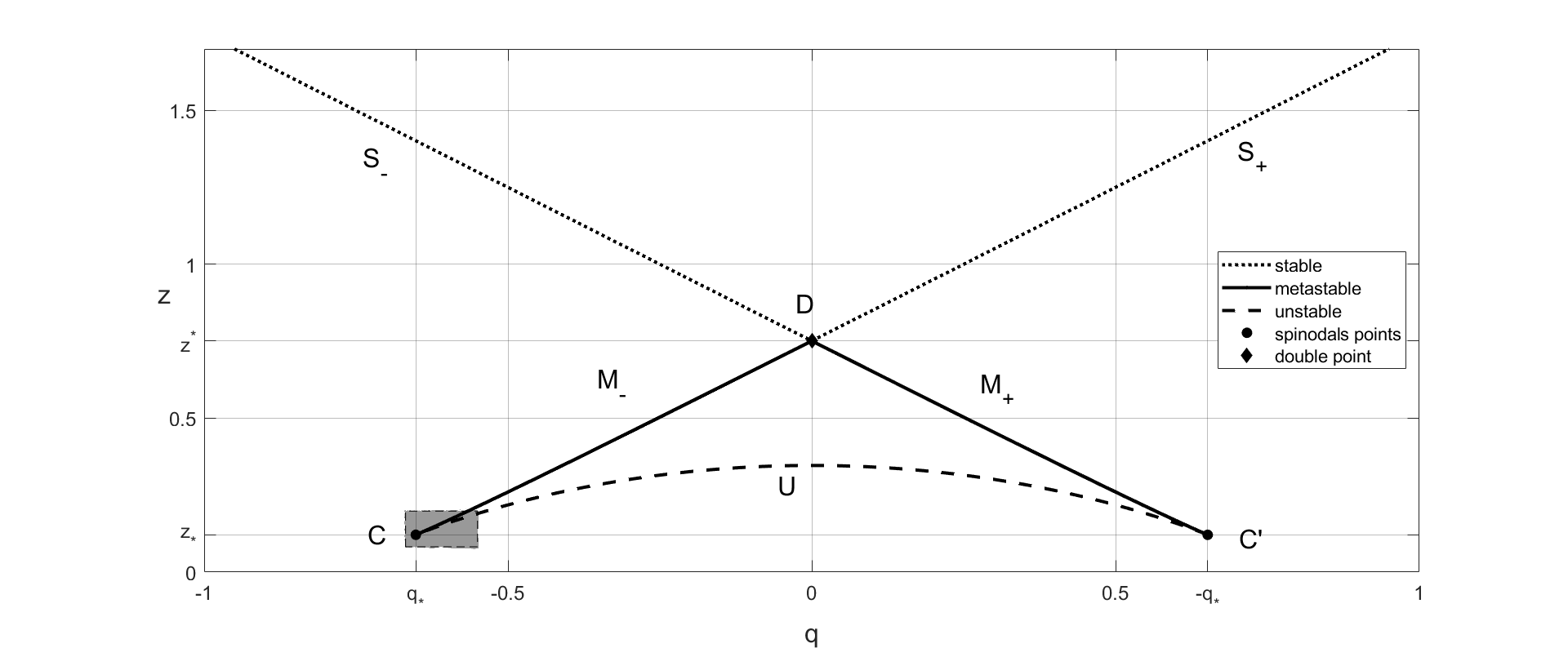}
\centering
\caption{Front $\Sigma$. Projection of Legendrian submanifold $\Lambda$ to
  the $(z,q)$-plane.}
\end{figure}

Now we are ready to outline the main findings of the present paper.

First, {\it  we design a contact Hamiltonian which roughly speaking
  ``imitates'' Glauber dynamics,}
see Section \ref{sec-design}. We refer the reader to \cite{Goto0,Goto,EP} for earlier steps in this direction.  Since the magnetic field is constant, such a Hamiltonian necessarily has a form $H(z,q)$ (i.e., it is independent of the magnetization $p$), and moreover since we wish to keep
$\Lambda$ as an equilibrium submanifold, $H(z,q)$ must vanish on the front $\Sigma$.
Let us emphasize that $H$ is defined near the spinodal point
$(z_*,q_* )$.
An important feature we wish to keep is that the unstable and metastable
equilibria of Glauber dynamics (cf. Convention above)
are unstable and stable, respectively, for the contact dynamics generated
by $H$.

Second, after designing such a Hamiltonian, {\it we prove the existence of the relaxation time
for the contact dynamics  provided $q \neq q_*$, and calculate
it}, see Section \ref{sec-contdyn}. It turns out (see Theorem \ref{thm-main-main} in Section \ref{sec-relax}) that the dependence of $\tau$ on $q$ is different
in the Glauber ($\sim (q-q_*)^{-1/2}$) and in the contact ($\sim (q-q_*)^{-3/2}$) cases, which highlights a subtle inconsistency between the two models; furthermore, we prove that this inconsistency is inevitable.

Third, we show that for $q=q_*$, {\it both in the Glauber and in the contact cases, the quantity $p(t)-p_*$ decays according to the {\it power laws} $\sim t^{-1}$}, see Theorem \ref{thm-main-main-2} in Section \ref{sec-power}. In the
contact case,  we are able to achieve such a law by a proper choice of the contact Hamiltonian.

Let us pass to precise formulations.

\section{Designing a contact Hamiltonian}\label{sec-design}

\begin{prop}\label{prop-1}  There exists a $C^{\infty}$-smooth function $H(z,q)$
defined in a neighbourhood $\mathcal{W}$ of $(q_*,z_*)$
with the following properties:
\begin{itemize}
\item The zero level $\{H=0\}$ coincides with $\Sigma \cap \mathcal{W}$;
\item The derivative $\partial H/\partial z (z,q)$ is strictly negative on $M_- \cap \mathcal{W}$ and strictly positive on $U \cap \mathcal{W}$.
\end{itemize}
\end{prop}

\medskip
\noindent
{\bf Proof:}
Let $(z_*,q_*,  p_*)$
be one of two spinodal points, with $p_* = \sqrt{1- (b\beta)^{-1}}$.
Recall that $q'(p_*)=0$. Denote $\theta:= q''(p_*) >0$.

Introduce a new parameter $P=p-p_*$ along the equilibrium submanifold $\Lambda$,
and make a change of variables:
$$
Z(P) = z(P+p_*) - p_*q(P+p_*)-(z_*-p_*q_*)\;,\qquad
Q(P) =q(P+p_*)-q_*\;.
$$
Clearly $Z(0)=Q(0)=0$.
Since $\Lambda$ is Legendrian, $z' = pq'$ (where by prime we denote the derivative with respect to $p$, and, with a slight abuse of notation, also with respect to $P$).
This yields $z'' = q' + pq''$ and $z''' = 2q''+pq'''$.

We have
\begin{align*}
Z'(0)   &= z'(p_*) - p_*q'(p_*)=0\;,\\
Z''(0)  &= q'(p_*) +p_*q''(p_*)- p_*q''(p_*)=0\;,\\
Z'''(0) &= 2q''(p_*) + p_*q'''(p_*) - p_*q'''(p_*) =2\theta\;.
\end{align*}
It follows
\begin{equation}\label{eq-Z}
  Z(P) = \frac{1}{3} \theta P^3 + O(P^4)\;.
\end{equation}
Similarly, $Q'(0) = 0$, $Q''(0) = \theta$, so
\begin{equation}\label{eq-Q}
  Q(P) = \frac{1}{2}\theta P^2 + O(P^3)\;.
\end{equation}

Now we use a standard trick of singularity theory (see \cite{W}, Lemma 2.3.1).
Notice that $P=0$ is a  non-degenerate critical point
of $Q\left(P\right)$ and by the Morse lemma there exists a local coordinate $s$ near $P=0$
such that $P = \sqrt{\theta/2}s + O(s^2)$ and $Q\left(s\right)=s^{2}$.
Rewrite $Z\left(P\right)$ with the parameter $s$:
$$
Z\left(s\right)=\sum_{n=3}^{\infty}\alpha_{n}s^{n}\;.
$$
Let $\phi_{0}\left(Q\right):=\sum_{n=2}^{\infty}\alpha_{2n}Q^{n}$ and
$\phi_{1}\left(Q\right):=\sum_{n=0}^{\infty}\alpha_{2n+3}Q^{n}$.
Since $Q\left(s\right)=s^{2}$,
we have that
$$
Z=\phi_{0}\left(Q\right)+sQ\phi_{1}\left(Q\right)\;,
$$
and
$$
s^{2}Q^{2}\phi_{1}\left(Q\right)^{2}=\left(Z-\phi_{0}\left(Q\right)\right)^{2}\;.
$$
Thus, since $Q(s)=s^2$,
the front near the spinodal point is given by the equation
$$
-\left(Z-\phi_{0}\left(Q\right)\right)^{2}+Q^{3}\phi_{1}\left(Q\right)^{2}
=0\;.
$$
Choose the Hamiltonian in a neighborhood of the spinodal point as
\begin{equation}\label{eq-H-full}
  H\left(Z,Q\right)
  =\left(1-aQ\right)\left(-\left(Z-\phi_{0}\left(Q\right)\right)^{2}
+Q^{3}\phi_{1}\left(Q\right)^{2}\right)\;,
\end{equation}
where $a$ is a real parameter.
Choose $Z$ and $Q$ sufficiently small, and also
  the parameter $a$ sufficiently small.
\newline
One readily checks that $H$ is as required.
\qed

The parameter $a$  will be used later in Section \ref{sec-power} to
adjust the relaxation rate at the spinodal point.

For future use, let us record the following relations:
\begin{equation}
\label{eq-future}
\phi_0(Q) = O(P^4)\;,\qquad
\phi_0'(Q) = O(P^2)\;.
\end{equation}

\section{Contact dynamics}\label{sec-contdyn}

Introduce new coordinates $(Z,Q,P)$ near the spinodal point $( z_*, q_*,p_*)$ by
$$
Z=z - p_*q-(z_*-p_*q_*)\;,\qquad
Q=q-q_*\;,\qquad
P= p-p_*\;.
$$
In these coordinates
the contact form $dz-pdq$ equals $dZ-PdQ$,  and the dynamics
generated by a contact Hamiltonian $H(Z,Q,P)$ is governed by the system
of ordinary differential equations
\begin{equation*}
\begin{cases}
\dot{Z} = H -P\frac{\partial H}{\partial P}  \\
\dot{Q} = - \frac{\partial H}{\partial P}\\
\dot{P} =P\frac{\partial H}{\partial Z} + \frac{\partial H}{\partial Q}\\
\end{cases}
\end{equation*}
In our setting the magnetic field $Q$ is constant and
is considered as a parameter.
Therefore, in view of the second equation,
we are interested in Hamiltonians $H$ depending only on $Z$ and $Q$.
In this case the above system simplifies to a triangular one,
\begin{equation}\label{eq-main-system}
\begin{cases}
\dot{Z} = H \\
\dot{P} =P\frac{\partial H}{\partial Z} + \frac{\partial H}{\partial Q}
\end{cases}
\end{equation}
(We refer the reader, for example, to \cite{EP} for a study of Hamiltonians depending on all the variables $Z,Q$, and $P$ in the context of a non equilibrium Ising model.)
Take any Hamiltonian as in Proposition \ref{prop-1}.
The front $\Sigma \cap \mathcal{W}$
is given by $Z=Z(P), Q = Q(P)$ with $P= dZ/dQ$. Recall that $H(Z(P),Q(P)) =0$ for all $P$.
Differentiating by $P$, we get that
\begin{equation}\label{eq-R}
P\frac{\partial H}{\partial Z} (Z(P),Q(P)) + \frac{\partial H}{\partial Q} (Z(P),Q(P))=0\;.
\end{equation}
(Note that this means that the right hand side of the second equation in \eqref{eq-main-system} vanishes, as it should be at the equilibrium point.)

Pick $P_{\infty}$ close to $0$, and put $Z_{\infty} = Z(P_{\infty})$, $Q_{\infty}= Q(P_{\infty})$.
Introduce the function
\begin{equation}\label{eq-R-1}
R(Z) := P_{\infty}\frac{\partial H}{\partial Z} (Z,Q_{\infty})  + \frac{\partial H}{\partial Q} (Z,Q_{\infty})\;.
\end{equation}
By \eqref{eq-R} we have that $R(Z_{\infty})=0$.

\begin{thm}\label{thm-main}
Assume that
\begin{equation} \label{eq-gamma}
  \frac{\partial H}{\partial Z} (Z_{\infty},Q_{\infty})
  = -\gamma < 0\;,
\end{equation}
and
\begin{equation} \label{eq-ineq-delta}
\frac{dR}{dZ} (Z_{\infty})= \delta \neq 0\;.
\end{equation}
Then for every initial condition $(Z(0),Q_{\infty},P(0))$ in a sufficiently small neighbourhood of $(Z_{\infty},Q_{\infty}, P_{\infty})$ with
$P(0) \neq P_{\infty}$ we have
\begin{equation}
\label{eq-relax-contact}
\lim_{t \to +\infty} \frac{\ln |P(t)-P_{\infty}|}{t} = -\gamma\;.
\end{equation}
\end{thm}

\medskip\noindent{\bf Proof:}
{\sc Case I:} Assume first that $Z(0) >  Z_{\infty}$.  Thus
\begin{equation} \label{eq-neq-vsp}
H(Z(0),Q_\infty) \neq 0\;.
\end{equation}
Let $Z(t)$ be the solution of the first equation of
system \eqref{eq-main-system}.
Observe that $Z(t) \to Z_{\infty}$, and by L'H\^{o}pital's rule
\begin{equation}\label{eq-H-as}
\lim_{t \to +\infty} \frac{\ln (-H(Z(t),Q_{\infty}))}{t} = \lim_{t \to +\infty} \frac{\partial H/\partial Z (Z(t),Q_{\infty}) \cdot H(Z(t),Q_{\infty})}{H(Z(t),Q_{\infty})} = -\gamma\;.
\end{equation}
Substitute $Z(t)$ into the second equation of \eqref{eq-main-system}, and rewrite it as
$$
\frac{d}{dt}(P-P_{\infty})
= (P-P_{\infty})\cdot \frac{\partial H}{\partial Z} (Z(t),Q_{\infty})  + R(Z(t))\;.
$$
Solving it, we get
\begin{equation}\label{eq-P-sol}
  P(t)-P_{\infty} = H(Z(t),  Q_{\infty}) \cdot I(t),\;\;\text{where}\;\;
  I(t):=  \frac{P(0)-P_{\infty}}{H(Z(0),  Q_\infty )}
  + \int_0^t \frac{R(Z(s))}{H(Z(s),  Q_\infty )}ds\;.
\end{equation}

Let us note that since $\dot{Z} = H$, inequality \eqref{eq-neq-vsp} yields
$H(Z(s),Q_\infty) \neq 0$ for all $s$.
Therefore, we can apply L'H\^{o}pital's rule
in combination with  \eqref{eq-gamma} and \eqref{eq-ineq-delta} and get that
$\lim_{s \to +\infty} \frac{R(Z(s))}{H(Z(s), Q_\infty )} = - \delta/\gamma \neq 0$.
Therefore, there exists $0 < c_1 < c_2$ such that for all $t$ large enough
\begin{equation} \label{eq-est}
c_1 t \leq |I(t)| \leq c_2t\;.
\end{equation}
Thus,
$$
\lim_{t \to +\infty} \frac{\ln |P(t)-P_{\infty}|}{t} = \lim_{t \to +\infty} \frac{\ln |H(Z(t), Q_\infty )|}{t}
+ \lim_{t \to +\infty} \frac{\ln |I(t)|}{t}\;,
$$
where the first term  on the right hand side equals $-\gamma$ by \eqref{eq-H-as},
and the second term vanishes by \eqref{eq-est}.
This yields the theorem if $Z(0) >  Z_{\infty}$.
The case $Z(0) <  Z_{\infty}$ is analogous.

\medskip\noindent
{\sc Case II:} Assume now $Z(0)=Z_{\infty}$. Then
$$
P(t) - P_{\infty} = e^{-\gamma t} (P(0) - P_{\infty})\;,
$$
and \eqref{eq-relax-contact} follows immediately.
This completes the proof.
\qed

\begin{rem}\label{rem-product}
  A direct calculation shows that the Hamiltonian $H(Z,Q)$ given
  by \eqref{eq-H-full}
  satisfies assumptions of Proposition \ref{prop-1} as well as \eqref{eq-gamma}
  and \eqref{eq-ineq-delta}. Therefore,
  the conclusion of Theorem \ref{thm-main} holds for $H$.
  A slightly more involved argument,
  which we leave to the reader,
  shows that the same is true for the Hamiltonian $F(Z,Q)H(Z,Q)$,
  where $F(Z,Q)$ is any positive smooth function defined
  in a neighbourhood of $0$.
\end{rem}

\begin{rem}\label{rem-GrobHart}
Theorem \ref{thm-main} readily follows from an enhanced version of the
Grobman-Hartman classical theorem \cite{N}, which is applicable in a much more
general situation. For the sake of completeness, we presented an elementary direct argument
working in our specific situation.
\end{rem}

\section{Comparison of the relaxation times} \label{sec-relax}

Now we are ready to formulate our first main result,
where $Q=q-q_{*}$ is considered as a small parameter.

\begin{thm} [Main Theorem-1] \label{thm-main-main}
$\;$
\begin{itemize}
\item[{(i)}] {\bf (Contact relaxation time)} The relaxation time $\tau_{\cont}$ of the Hamiltonian \eqref{eq-H-full}
satisfies the scaling law
\begin{equation} \label{eq-cont-relax}
 \tau_{\cont} \sim Q^{-3/2}\;.
\end{equation}
\item[{(ii)}] {\bf (No-Go theorem)} For every contact Hamiltonian $H(Z,Q)$ vanishing
on the front near the spinodal point and satisfying assumptions \eqref{eq-gamma}
and \eqref{eq-ineq-delta}, the relaxation time $\tau_{\cont}$ is at least
$\text{const} \cdot Q^{-3/2}$.
\item[{(iii)}] {\bf (Glauber relaxation time)} The relaxation time
  $\tau_{\glaub}$
of Glauber dynamics given by equation \eqref{eq-kubo} satisfies the scaling law
 \begin{equation}\label{eq-glaub-relax}
\tau_{\glaub} \sim Q^{-1/2}\;.
\end{equation}
In particular, by (ii) this scaling law cannot be modeled by contact dynamics.
\end{itemize}
\end{thm}

\medskip
\noindent
{\bf Proof of (i):}
We combine the results of Section \ref{sec-design}
with Theorem \ref{thm-main}. Recall that  by \eqref{eq-H-full}
$$
 H\left(Z,Q\right)=\left(1-aQ\right)\left(-\left(Z-\phi_{0}\left(Q\right)\right)^{2}
 +Q^{3}\phi_{1}\left(Q\right)^{2}\right)\;,
 $$
where $Q=q-q_* >0$, and $a$ is a real parameter
which we choose sufficiently small. We calculate that
at a point $Z=Z(P), Q=Q(P)$ of the front
$$
\frac{\partial H}{\partial Z}(Z,Q)
= -2(1-aQ)(Z-\phi_{0}(Q))
= -2\left(\frac{\theta}{3}\right)P^3 +O(P^4)\;,
$$
where the last equality follows from by \eqref{eq-Z}, \eqref{eq-Q}, \eqref{eq-future}.
Since $P_{\infty} > 0$ (meaning that the point we are working with is metastable)
and $P_{\infty}$ is chosen to be close to $0$, we have
$$
\gamma:= \frac{\partial H}{\partial Z}(Z_{\infty},Q_{\infty}) < 0\;,
$$
yielding assumption \eqref{eq-gamma}.
We record that
\begin{equation}\label{eq-gamma-as}
\gamma \sim Z_{\infty} \sim Q^{3/2}\;.
\end{equation}
Furthermore,
$$
\frac{dR}{dZ}(Z_{\infty})
= -2(1-aQ_{\infty})(P_{\infty}-\phi_{0}'(Q_{\infty}))
+ 2a(Z_{\infty} -\phi_0(Q_{\infty}))
=-2P_{\infty} + O(P_{\infty}^2)\;.
$$
Thus for $P_{\infty} > 0$ sufficiently close to $0$ we have by \eqref{eq-Z}, \eqref{eq-Q}, \eqref{eq-future} that $dR/dZ(Z_{\infty}) \neq 0$, and
hence assumption \eqref{eq-ineq-delta} holds.
Therefore, we can apply Theorem \ref{thm-main}. By formula \eqref{eq-gamma-as}, we get
the relaxation time $\tau_{\cont} \sim Q^{-3/2}$, as required.
\qed

\medskip\noindent
{\bf Proof of (ii):} Let $H(Z,Q)$ be any Hamiltonian vanishing
on the front $\Sigma$ in a neighbourhood of the spinodal point and satisfying
assumptions \eqref{eq-gamma} and \eqref{eq-ineq-delta}.
Our task is to estimate from below the relaxation time
$\tau_{\cont}$. By formulas \eqref{eq-relax} and \eqref{eq-relax-contact}
$$
\tau_{\cont} = \gamma^{-1}, \qquad 
\gamma = - \frac{\partial H}{\partial Z}(Z_\infty,Q_\infty)\;.
$$

Observe that $\partial H/\partial Z(0,0)$ and $\partial H/\partial Q (0,0)$
vanish as $\Sigma$ has a singularity at the origin. Write
$$H(Z,Q) = \sum_{m\geq 0,n\geq 0,m+n \geq 2} r_{mn}Z^mQ^n\;.$$
We claim that $r_{02}=r_{11}=0$. Indeed, by
\eqref{eq-Z} $Z(P) = \theta P^3/3 + O(P^4)$ and
by \eqref{eq-Q} $Q(P) = \theta P^2/2 + O(P^3)$.
Look at the expansion in $P$ of the equation
\begin{equation}\label{eq-defin}
H(Z(P),Q(P))=0\;.
\end{equation}
Call $3m+2n$ a {\it weight} of the monomial $Z^mQ^n$. If $r_{02}\neq 0$,
$Q^2$ is the unique monomial of the minimal weight $4$. This contradicts to \eqref{eq-defin},
and hence $r_{02} = 0$. If $r_{11}\neq 0$, $ZQ$ is the unique monomial of  the minimal weight $5$,
which again contradicts to \eqref{eq-defin}. Thus, $r_{11}=0$, and the claim follows.

Since $r_{11}=0$,
$$
\frac{\partial H}{\partial Z} (Z,Q) = 2
r_{20}Z + r_{12}Q^2 + \rho(Z,Q)\;,
$$
where $\rho$ consists of monomials of higher weight. It follows that
\begin{equation}\label{eq- H-est}
  \left| \frac{\partial H}{\partial Z} (Z,Q)\right|
  = O(|Z| + Q^2)\;.
\end{equation}
Taking into account that $|Z_\infty| \sim Q_\infty^{3/2}$,
we get that
$$
\gamma \leq \text{const} \cdot Q_\infty^{3/2}\;.
$$
This yields
$$
\tau_{\cont} \geq  \text{const} \cdot Q_\infty^{-3/2}\;,
$$
as required.
\qed

\medskip\noindent
{\bf Proof of (iii):}
Now let us elaborate on the relaxation time for Glauber dynamics.
The Glauber equation has the form
$$
\dot{p} = u(p) = -p + \tanh \beta(q+bp)\;.
$$
We take the value of $q$ of the form $q_* + Q$, $Q >0$ and
look at the metastable equilibrium $p_{\infty}= p_*+P_{\infty}$.
Applying L'H\^{o}pital's rule as in the first step of the proof
of Theorem \ref{thm-main} we get that the relaxation time is well defined
and equals $-(u'(p_{\infty}))^{-1}$. We calculate,
taking into account that $p_*^2 =1 -1/(b\beta)$
and $u(p_{\infty})=0$, that
$$
u'(p_{\infty}) = -1 + b\beta (1- p_{\infty}^2) = -2b\beta p_* P_{\infty} + O(P_{\infty}^2)\;.
$$
Recalling that near the spinodal point $Q \sim P^2$ (see \eqref{eq-Q}),  we
see that the relaxation time equals
$\tau_{\glaub} \sim Q^{-1/2}$, as required.
\qed

\section{Power law at the spinodal point}\label{sec-power}

Write $P(t)=p(t)-p_*$, where $p_*$ is the magnetization at the spinodal point.
In the contact case, assume that the parameter $a$ in formula \eqref{eq-H-full}
does not vanish.

\begin{thm} [Main Theorem-2] \label{thm-main-main-2} Both in the contact case and in the Glauber case the relaxation dynamics of the magnetization at the spinodal point is given by the power law
\begin{equation}\label{eq-power-2}
P \sim t^{-1},\qquad \; t \gg 1\;.
\end{equation}
\end{thm}

\medskip\noindent{\bf Proof:}
We start with the contact case, assuming that $a \neq 0$.
Our objective is to solve the contact Hamiltonian system with the initial conditions
$$
Q(0)=0\;,\qquad
Z(0) = Z_{0} >0\;,\qquad
P(0) = P_{0} >0\;.
$$
Note that $Q(t)=Q(0)=0$.
We have $\dot{Z}=H\left(Z,Q\right)=-\left(Z-\phi_{0}\left(0\right)\right)^{2}$.
Recall that $\phi_{0}\left(Q\right)=\sum_{n=2}^{\infty}\alpha_{2n}Q^{n}$
and $\phi_{0}\left(0\right)=0$. Thus, our equation reads $\dot{Z}=-Z^{2}$, so
$$
Z(t) = \left(t+\frac{1}{Z_{0}}\right)^{-1}\;.
$$
Next,
\begin{align*}
\dot{P} & =P\frac{\partial H}{\partial Z}+\frac{\partial H}{\partial Q}\\
 & =-2P\left(1-aQ\right)\left(Z-\phi_{0}\left(Q\right)\right)-a\left(-\left(Z-\phi_{0}\left(Q\right)\right)^{2}+Q^{3}\phi_{1}\left(Q\right)^{2}\right)\\
 & +\left(1-aQ\right)\left(-2\phi_{0}^{'}\left(Q\right)\left(\phi_{0}\left(Q\right)-Z\right)+3Q^{2}\phi_{1}\left(Q\right)^{2}+2Q^{3}\phi_{1}^{'}\left(Q\right)\phi_{1}\left(Q\right)\right)
\end{align*}
For $Q=0$, we have $\phi_{0}^{'}\left(Q\right)=0$, so
$$
\dot{P} =  -2ZP +aZ^{2}= -2\left(t+\frac{1}{Z_{0}}\right)^{-1}P
+ a\left(t+\frac{1}{Z_{0}}\right)^{-2}\;.
$$
This is a linear non-homogeneous equation. Its solution is given by
$$
  P(t)= C\left(t+\frac{1}{Z_{0}}\right)^{-2}
  + a\left(t+\frac{1}{Z_{0}}\right)^{-1}\;,
$$
where the first term is the general solution of the homogeneous equation,
and the second term is a special solution of the non-homogeneous equation.
Incorporating the initial condition, we get
\begin{equation} \label{eq-const}
C= \frac{(P_{0}-aZ_{0})}{Z_{0}^{2}}\;.
\end{equation}

We rewrite this in the original coordinates as
\begin{equation}\label{eq-power}
  p(t) = p_* + C\left(t+\frac{1}{Z_0}\right)^{-2}
  + a\left(t+\frac{1}{Z_0}\right)^{-1}\;,
\end{equation}
where
\begin{equation}\label{eq-power1}
  Z_0 = z(0)-z_*\;,\qquad
  C= \frac{p(0)-p_{*}-aZ_{0}}{Z_0^2}\;.
\end{equation}
Thus $P(t) \sim t^{-1}$ when $t \to +\infty$, as required.

\medskip

In Glauber dynamics, the equation has the form $\dot{p} = u(p)$.
An easy calculations shows that
at the spinodal point $u(p_*) = 0, u'(p_*)=0, u''(p_*) = -2\eta < 0$,
where $\eta=b\beta\,p_{*}>0$.
Approximating the dynamics by the equation $\dot{P} = - \eta P^2$,
we readily deduce the power law \eqref{eq-power-2}.
Indeed, put $U(P) = u(p+p_*)$,  and write the function $U$ in the form
$$
U(P) = -\eta P^2 (1 + P\psi(P))\;,\qquad \eta >0\;,
$$
where $P$ lies in a sufficiently
small neighborhood $(-\epsilon,\epsilon)$ of $0$ (see below),
and $\psi$ in this neighbourhood satisfies a bound $|\psi| \leq c$.
Put $V(P)= -\eta P^2$. Let $P(t)$ be the solution of the equation
\begin{equation} \label{eq-ODEU}
\dot{P}=U(P)
\end{equation}
with an initial condition $P(0) >0$. Clearly, if $P(0)$ is sufficiently
close to $0$, we have $P(0) > P(t) >0$ for all $t >0$, and
\begin{equation} \label{eq-pinf}
P(t) \to 0\;.
\end{equation}
Put
$$
J := \int_{P(t)}^{P(0)} \left(\frac{1}{V(s)} - \frac{1}{U(s)}\right)ds\;.
$$
We calculate
$$
J = -\int_{P(t)}^{P(0)} \frac{\psi(s)}{\eta s(1+s\psi(s))} ds\;.
$$
Assume now that $\epsilon < 1/(2c)$.
Then $1+s\psi(s) \geq 1-c\epsilon \geq 1/2$. We estimate
$$
\bigg| \frac{\psi(s)}{\eta s(1+s\psi(s))}\bigg|  \leq \frac{c_1}{s}\;,
$$
with $c_1 = 2c/\eta$. It follows that
\begin{equation}\label{eq-est-1}
|J| \leq c_1|\ln P(t) - \ln P(0)| \leq c_1 |\ln P(t)| + c_1|\ln P(0)|\;.
\end{equation}
From the equation \eqref{eq-ODEU} we get
$$
dt
= \frac{dP}{U}
= \frac{dP}{V} +\bigg( \frac{dP}{U}-\frac{dP}{V}\bigg)\;.
$$
Thus
$$
t = \frac{1}{\eta P(t)} - \frac{1}{\eta P(0)} - J\;.
$$
By \eqref{eq-est-1} and \eqref{eq-pinf},
we get
$$
tP(t) \to \frac{1}{\eta},\qquad \; t \to +\infty\;.
$$
Thus
\begin{equation} \label{eq-final-p}
 P(t) = \frac{1}{\eta t} + o(t^{-1})\;,
\end{equation}
as required.
This completes the proof.
\qed

\section{Conclusion and open problems}
This paper contributes to
a description of nonequilibrium dynamics in the presence of a phase transition.
We have designed a contact geometric model of a nonequilibrium thermodynamic
system in the low temperature regime,
where this system describes the
time-development of the magnetization of
the Ising model with mean field type interactions and shows
a first order phase transition.
This system exhibits relaxation towards equilibrium states,
in agreement with a fundamental model of nonequilibrium thermodynamics, 
Glauber dynamics. The merit of the contact
dynamical system, defined in the thermodynamic phase space,
is that in contrast to Glauber dynamics it automatically
preserves the kernel of the Gibbs fundamental form $dz-pdq$,
i.e., preserves the fundamental thermodynamic relation.
Note that this does not automatically
  induce
  the preservation of the probability distribution function in phase space
  and vice versa. Meanwhile, it is possible to discuss relations between such a
  distribution function and a contact form
  \cite{Bravetti}.
    At the same time we have proved
a No-Go theorem stating that in a neighbourhood of the spinodal point
the relaxation time of the contact system towards a metastable equilibrium
is always
larger than the one of the Glauber system, independently of
the choice of the contact Hamiltonian.
Which of the two models provides a more accurate description of relaxation
processes remains an open problem. In particular, it would be interesting
to make a comparison with the metastable behavior of the Markov process modeling
the Ising chain relaxation in the framework of the Curie-Weiss model \cite[Chapter 13]{BH}.

Feasibility of an emulation of relaxation processes in a given region of the thermodynamic
phase space by using contact flows depends on the postulated dynamical behavior near metastable
equilibria. In the present paper we assumed that the metastable equilibria are stable for the flow.
This assumption, however, has limitations. For instance, it  would prevent us from designing the desired contact Hamiltonian in a neighbourhood of the double point $D$ of the front, see Figure \ref{ising-fig2}. Indeed, look at the regions bounded by the stable and metastable branches, and recall that the contact evolution of the energy $z$ is given by $\dot{z}=H$.
Since the stable branch $S_{-}$ lies above $M_{-}$, $\partial H/\partial z$
is necessarily negative at $S_{-}$, and hence $M_{-}$ become unstable.
We refer the reader to \cite{Goto} where metastable states were
treated as dynamically unstable ones.

While in the present paper we have focused on the contact dynamics
in a three dimensional thermodynamic phase space,
we expect that our methodology extends to higher dimensional models.
This requires a more systematic
procedure of designing contact Hamiltonians involving more sophisticated
tools of singularity theory.

Let us mention also that contact geometry and contact dynamics form
just one of several facets of relations
between thermodynamics and differential geometry. In particular,
in the present paper we have  not touched Riemannian geometry of the
thermodynamic phase space. It would be interesting to explore its benefits
for modeling relaxation processes of
nonequilibrium  thermodynamics
near spinodal points.

The approach of this paper should be applicable to
emulation of relaxation processes in other thermodynamic models,
both in terms of designing a suitable contact Hamiltonian,
and understanding limitations of the contact geometric framework.
Such models include, among others,  black hole physics and control systems.

We close this paper with two open problems. 

\subsection{Open problem I: Glauber equation in higher dimensions}
Consider an ordinary differential equation
\begin{equation}\label{eq-kubo-11}
\dot{p} = -p + \phi'(q+bp)\;,
\end{equation} where $b$ is a real parameter, $p,q \in \R^n$, and  $\phi$ is a smooth function on $\R^n$. Here $q = \text{const}$, and we write $\phi'$ for the gradient of $\phi$.
On the one hand, equation \eqref{eq-kubo-11} is a direct generalization of the Glauber equation
\eqref{eq-kubo}. On the other hand, it is closely related to the description of nonequilibrium thermodynamics in terms of affinities and fluxes, see \cite{Haslach}.
Let us explain this in more detail.

To this end, make a change of variables
$$
x = p+b^{-1}q,\qquad
y= q,\qquad
w = z+b^{-1}\frac{q^2}{2}\;,
$$
so that the contact form is given by
$$
dz-pdq=dw-xdy\;.
$$
In the new coordinates equation \eqref{eq-kubo-11} reads
\begin{equation}
\label{eq-Haslach-10}
\dot{x} = b^{-1}(y - bx+ b\phi'(bx))\;.
\end{equation}
Set $c=b^{-1}$, $\psi(x):= -bx^2/2+ \phi(bx)$ and
define the generalized energy function in the sense
of Haslach,
$$
E(x,y) = xy + \psi(x)\;.
$$
With this language, equation \eqref{eq-Haslach-10}, i.e., the generalized Glauber equation
written in the new coordinates, has the form
\begin{equation}\label{eq-Haslach-2}
\dot{x} = c \frac{\partial E}{\partial x} (x,y),\qquad \; y = \text{const}\;.
\end{equation}
This is equivalent to the Haslach gradient flow equation for the affinities \cite[equation (10)]{Haslach}, where the latter are given by $X_i = \partial E/\partial x_i$, and the
the gradient is understood with respect to the push-forward of the
Euclidean metric on $\R^n(x)$ to the space of affinities under the map
$x \mapsto X$.
Here we tacitly assume that this map is a local diffeomorphism.

With this motivation at hand, we address the following problem.

\begin{problem} Extend the results of the present paper to equation
\eqref{eq-kubo-11} in arbitrary dimension.  More precisely, we propose to look
at the neighbourhoods of the singular points of the front projection of the equilibrium Legendrian submanifold
$$
\left\{\ p = \phi'(q+bp),\quad
z= -\frac{p^2}{2}+ b^{-1}\phi(q+bp)
\ \right\} \subset \R^{2n+1} \;,
$$
imitate the dynamics given by \eqref{eq-kubo-11}
by a contact Hamiltonian flow, and explore the limitations.
\end{problem}

\medskip
We expect that while the general strategy should follow the lines
of the present paper, the analysis of singularities should be
more sophisticated.

\subsection{Open problem II: a microscopic approach to contact dynamics}
Interestingly enough, the question about the power law at the spinodal
points (cf. Section \ref{sec-power} above)  was addressed in the literature \cite{A,L}, albeit in a different context of a Monte-Carlo
type dynamics discussed in \cite{Z}. These papers focus on the dynamics
corresponding to the arrival of the system at the metastable state (see \cite[Section 5]{A} and \cite[Section 1]{L}), and in particular on the corresponding scaling behavior.

Paper \cite[p.37, Section 4.4]{Z} mentions that
``it has long been challenging whether stochastic dynamics is equivalent to the fundamental deterministic dynamics, and vice versa.'' Performing such a comparison in our situation
is an open and apparently difficult mathematical problem. The first step would be, following  a proposal by S.~Shlosman discussed in \cite{EP}, to derive rigorously ODE \eqref{eq-kubo}  in an appropriate thermodynamic limit of the Curie-Weiss model. If this succeeds, the next step would be to derive the contact dynamics generated by Hamiltonian \eqref{eq-H-full} starting from the microscopic set up. 

Let us perform a naive comparison of our results on the contact dynamics with the findings of \cite{L} and \cite{A}.

In  \cite{L}, the quantity $P$ is called $\Delta m$ and
is introduced after formula (13);  it is calculated
at the spinodal point right after formula (25) as $\sim t^{-0.98}$.

In \cite{A},  the quantity $P$ is called  $\Delta m$ and
is introduced before formula (27);
it is calculated at the spinodal point right after formula (31)
as $\sim t^{-1}$, for sufficiently large values of time $t$.

These results show a good agreement with our formula \eqref{eq-power-2}.
This can be considered as an argument in favor of the existence of a rigorous
microscopic approach to contact dynamics.

\subsection*{Acknowledgments}
The author S.G. was partially  supported by the JSPS (KAKENHI)
(Grant No. JP19K03635). The other authors, S.L and L.P.
were partially supported by the Israel Science Foundation grant 1102/20.
The authors S.G. and L.P. thank Minoru Koga at Nagoya University
for giving various suggestions and fruitful discussions on this study.
L.P. thanks Michail Zhitomirskii from the Technion for a consultation on
singularity theory.
We thank anonymous referees for useful comments.

\subsubsection*{Data availability statement}
No new data were created or analysed in this study.

\end{document}